\documentclass[12pt,a4paper] {article}
\usepackage{latexsym,amsmath,enumerate,amssymb,amsbsy,amsthm}
\usepackage{enumerate,verbatim,tikz}

\theoremstyle{plain}

\newtheorem{thm}[equation]{\bf Theorem}
\newtheorem{prop}[equation]{\bf Proposition}
\newtheorem{cor}[equation]{\bf Corollary}
\newtheorem{lemma}[equation]{\bf Lemma}

\theoremstyle{remark}

\theoremstyle{definition}

\numberwithin{equation}{section}

\renewcommand\footnotemark{}

\date{}

\begin{document}

    \title{Self-Conjugate-Reciprocal Irreducible Monic Factors of  $x^n-1$ over Finite Fields and Their Applications}

    \author{Arunwan Boripan, Somphong~Jitman and Patanee Udomkavanich}

    \thanks{A. Boripan is with the  Department of Mathematics and Computer Science,
    	Faculty of Science,  Chulalongkorn University,   Bangkok 10330,  Thailand  (email: boripan-arunwan@hotmail.com)}

    \thanks{S. Jitman (Corresponding Author)  is with the  Department of Mathematics, Faculty of Science,
    	Silpakorn University, Nakhon Pathom 73000,  Thailand
    	(email: sjitman@gmail.com).}

       \thanks{P. Udomkavanich is with the  Department of Mathematics and Computer Science,
       	Faculty of Science,  Chulalongkorn University,   Bangkok 10330,  Thailand  (email: pattanee.u@chula.ac.th)}
       
    \thanks{This research was  supported by the Thailand Research Fund and the Office of Higher Education Commission of Thailand under Research
    	Grant MRG6080012.}

    \maketitle

\begin{abstract}
Self-reciprocal  and self-conjugate-reciprocal  polynomials  over finite fields have been of interest due to their rich algebraic structures and wide applications.  Self-reciprocal irreducible monic factors of $x^n-1$ over finite fields and their applications  have been quite well studied.  In this paper, self-conjugate-reciprocal irreducible monic (SCRIM) factors of $x^n-1$ over finite fields of square order have been focused on. The characterization of such factors is given together the enumeration formula. In many cases,  recursive formulas for the number of SCRIM factors of $x^n-1$ have been given as well. As applications,  Hermitian complementary dual codes over finite fields and  Hermitian self-dual cyclic codes over finite chain rings of prime characteristic have been discussed. 

\smallskip


\end{abstract}
 
\section{Introduction}

A non-zero  polynomial  $f(x)$   over a finite field $\mathbb{F}_q$  whose constant term is a unit    is said to be  \emph{self-reciprocal}   if $f(x)$ equals its  {\em reciprocal polynomial} $f^*(x):=x^{\deg(f(x))}f(0)^{-1}f\left(\frac{1}{x}\right)$. A  polynomial is said to be  {\em self-reciprocal irreducible monic} (SRIM) if it is self-reciprocal, irreducible and monic.   Due to their rich algebraic structures and wide applications,  SRIM and self-reciprocal polynomials  over finite fields have been studied and applied in various branches of Mathematics and Engineering.  In \cite{HB1975},   SRIM polynomials have been  characterized up to their degrees.    The orders and the number  of  SRIM polynomials of a given degree  over finite fields have been determined in  \cite{YM2004}.    SRIM factors of $x^n-1$   were used for characterizing and enumerating Euclidean self-dual cyclic codes over finite fields in \cite{JLX2011} and  characterizing   Euclidean complementary dual cyclic codes over finite fields in \cite{YM}. Recently, the enumeration of  simple root Euclidean self-dual cyclic codes    over finite chain rings has been given in terms of SRIM factors of $x^n-1$ over finite fields in \cite{BSG2016}.

Here, we focus on an extension of  SRIM polynomials over finite fields. Over a finite field $\mathbb{F}_{q^2}$ of square order, the {\em conjugate} of a  polynomial $f(x)=\sum_{i=0}^n f_i x^i$ over $\mathbb{F}_{q^2}$ is defined to be  $\overline{f(x)}=   \overline{f_0}+\overline{f_1}x+\dots+\overline{f_n}x^n$, where 
$\bar{~} : \mathbb{F}_{q^2} \rightarrow  \mathbb{F}_{q^2}$ is the field automorphism given  by $\alpha\mapsto \alpha^{q}$ for all $\alpha\in \mathbb{F}_{q^2}$. A polynomial $f(x)$ over $\mathbb{F}_{q^2}$ (with $f(0)\ne 0$)  is said to be {\em self-conjugate-reciprocal} if $f(x)$  equals its  {\em conjugate-reciprocal polynomial} $f^\dagger(x):= \overline{f^*(x)}$.  
If, in addition, $f(x)$  is monic and irreducible, it is said to be \emph{self-conjugate-reciprocal irreducible monic} (SCRIM).  Characterization of monic irreducible polynomials over  $\mathbb{F}_{q^2}$ to be SCRIM have been  given together with the  enumeration  of SCRIM polynomials of a fixed  degree   in \cite{BJU2015}. Some properties of SCRIM factors of $x^n-1$ over  $\mathbb{F}_{ 2^{2l}}$  have been studied and used in the characterization and enumeration of  Hermitian self-dual cyclic codes in \cite{JLS2014}.

In this paper, we focus on SCRIM factors of $x^n-1$ over finite fields $\mathbb{F}_{q^2}$, where $q$ is an arbitrary prime power. The  characterization and enumeration of such polynomials are given in Section 2. In Section 3,   SCRIM factors of $x^n-1$ over finite fields $\mathbb{F}_{q^2}$ are applied in coding theory.  Precisely, Hermitian self-dual cyclic codes over finite chain rings of prime characteristic  and    Hermitian complementary dual cyclic codes are characterized and enumerated in terms of   SCRIM factors of $x^n-1$ over  $\mathbb{F}_{q^2}$.

\section{SCRIM Factors of  $x^n-1$ over Finite Fields}
In this section, the   SCRIM factors of $x^n-1$ over $\mathbb{F}_{q^2}$  are investigated in the case where $ q$ is an arbitrary   prime power and $n$ is a positive integer such that $\gcd(n,q)=1$.  The characterization of such factors is given in Subsection 3.1 and the enumeration is provided  in Subsection 3.2.

\subsection{Characterization of SCRIM Factors of  $x^n-1$ over $\mathbb{F}_{q^2}$}

Observe that  for every monic polynomial $f(x)$ in $\mathbb{F}_{q^2}[x]$ with $f(0)\ne 0$, we have  $\left(f^\dagger (x)\right)^\dagger =f(x)$. Therefore,   for a monic irreducible polynomial $f(x)\in \mathbb{F}_{q^2}[x]$, either $f(x)$ is SCRIM or $f(x)$ and $f^\dagger(x)$ form a   pair of distinct polynomials.  In the latter case, $f(x)$ and $f^\dagger(x)$ are called a {\em conjugate-reciprocal irreducible monic (CRIM) polynomial pair}.

Let $q$ be a prime power and let   $n$ be a positive integer such that $\gcd(g,n)=1$. Denote by $\Omega_{q^2,n}$ and  $\Lambda_{q^2,n}$ 
the set  of  SCRIM factors of $x^n-1$  and   the set of pairs of  CRIM polynomial pairs in the factorization of $x^n-1$ in $\mathbb{F}_{q^2}[x]$, respectively.  Then    $x^n-1$    can be factorized into  a product of irreducible monic polynomials  in  $\mathbb{F}_{q^2}[x]$ of the following form 
\begin{align}\label{eq-fistFactor}
x^n-1=\prod_{i=1}^{|\Omega_{q^2,n}|}  f_i(x)\prod_{j=1}^{|\Lambda_{q^2,n}|}\left(g_i(x)g_j^{\dagger} (x)\right),
\end{align}
where 
$ f_i(x)$  is a SCRIM polynomial and  $ g_j(x)$  and $   g_j^\dagger (x) $  are a CRIM polynomial pair
for all $1\leq i\leq |\Omega_{q^2,n}|$ and $1\leq j\leq |\Lambda_{q^2,n}|$.

For each coprime positive integers $i$ and $j$, the \textit{multiplicative order of $j$ modulo $i$}, denoted by  $\textrm{ord}_i{(j)}$,  is defined to be the smallest positive integer $s$ such that $j^s\equiv 1 \,\textrm{mod}\,i$. For a positive integer $i$ and nonnegative integer $s$, denote by   $2^s||i$ if $s$ is the largest integer such that $i$ is divisible by $2^s$, i.e., $2^s|i$ but $2^{s+1}\nmid i$.

For each $0\leq i<n$,  the \emph{cyclotomic  coset of $q^2$ modulo  $n$  containing $i$} is defined to be the set 
\begin{align*}
Cl_{q^2,n}(i) =\{iq^{2j} \, \textrm{mod}~ n \mid j =0,1,2,\dots\}.
\end{align*}
It is not difficult to see that  $Cl_{q^2,n}(i)  =\{iq^{2j} \, \textrm{mod n} \mid  0\leq j < \textrm{ord}_i{(q^2)}\}$ and $|Cl_{q^2,n}(i)  |=\textrm{ord}_i{(q^2)}$ .

For a primitive $n$th root of unity $\alpha$ in some extension field of $\mathbb{F}_{q^2}$ and for each    $0\leq i<n$,  it is well known (see \cite{LC2004}) that 
\[f(x)=\prod_{j\in Cl_{q^2,n}(i) }{(x-\alpha^j)}\]
is the minimal polynomial of $\alpha^i$ and it is a monic irreducible factor of $x^n-1$ over $\mathbb{F}_{q^2}$. 

%
%
%
%
%
%

We  have the following property.
\begin{lemma}[{\cite[Lemma 3.2]{BJU2015}}] \label{thm -qi}
    Let   $n$ be a positive integer such that $\gcd(q,n)=1$ and let $\alpha$ be a primitive $n$th root of unity. Let $0\leq i<n$ and let  $f(x)$ be the minimal polynomial of $\alpha^i$.   Then $f(x)$ is  SCRIM  if and only if $Cl_{q^2,n}(i) =Cl_{q^2,n}(-qi) $.
\end{lemma}

Using the  analysis similar to those in \cite[Section 4]{SJLU2015}, the set of  SCRIM factors of $x^n-1$ over $\mathbb{F}_{q^2}$  is 
\begin{align}
\Omega_{q^2,n}=\bigcup_{ d|n,\lambda(q,d)=1}\{ \prod_{j\in Cl_{q^2,n}(i)}{(x-\alpha^j)} \mid  0\leq i<n \text{ and  } {\rm ord}_i(q^2)=d\}
\end{align}  and  the set of pairs of  CRIM polynomial pairs in the factorization of $x^n-1$ over $\mathbb{F}_{q^2}$  is 
\begin{align} 
\Lambda_{q^2,n}=\bigcup_{ d|n,\lambda(q,d)=0}\{ (\prod_{j\in Cl_{q^2,n}(i)}{(x-\alpha^j)} ,& \prod_{j\in Cl_{q^2,n}(-i)}{(x-\alpha^j)} )\mid  0\leq i< n \notag \\
&\text{ and  } {\rm ord}_i(q^2)=d\},
\end{align}
where $\alpha$ is a primitive $n$th root of unity and  \[\lambda(q,i)=\begin{cases}
1 &\text{if there exists an odd positive integer }e \text{ such that } i|(q^e+1),\\
0 &\text{if } i\nmid (q^e+1)  \text{ for all  odd positive integers }e .
\end{cases}
\]
Hence, the number of  SCRIM factors of $x^n-1$ over $\mathbb{F}_{q^2}$  is 
\begin{align}\label{formulaO}
|\Omega_{q^2,n}|= \sum_{d|n}\lambda(q,d)\frac{\phi(d)}{\mathrm{ord}_d(q^2)}
\end{align}  and  the number of pairs of  CRIM polynomial pairs in the factorization of $x^n-1$ over $\mathbb{F}_{q^2}$  is 
\begin{align} 
|\Lambda_{q^2,n}|=\sum_{d|n}(1-\lambda(q,d))\frac{\phi(d)}{2\mathrm{ord}_d(q^2)}.
\end{align}
The number of monic irreducible factors of $x^n-1 $ over $\mathbb{F}_{q^2}$  is 
\[\sum_{d|n}\frac{\phi(d)}{\mathrm{ord}_d(q^2)}=|\Omega_{q^2,n}|+2|\Lambda_{q^2,n}|. \]

Next, we focus on the following extreme cases where $x-1$ is the only SCRIM factor of $x^n-1$ over $\mathbb{F}_{q^2}$  and where  all the factors of $x^n-1$ over $\mathbb{F}_{q^2}$   are  SCRIM.

\begin{thm}\label{order}
    Let $l$ be an odd prime  such that  $l\nmid q$. Then all monic irreducible factors of $x^{l}-1$ over $\mathbb{F}_{q^2}$ are SCRIM  if and only if $\textrm{ord}_l{(q^2)}$ is odd but  $\textrm{ord}_l{(q)}$ is even.
\end{thm}

\begin{proof}
    Assume that all monic irreducible factors of $x^{l}-1$ are SCRIM. Then $Cl_{q^2,l}(i)=Cl_{q^2,l}(-qi)$  for all  $0\leq i<l$ by Lemma \ref{thm -qi}.  Let $0< h < l$ and $v_l=\textrm{ord}_l{(q^2)}$.  Then  $Cl_{q^2,l}(h) =Cl_{q^2,l}(-qh)$.  Precisely, there exists an integer $ j\geq 0$ such that
    \begin{align}
    \label{th1e1}    h\equiv (-qh)(q^{2j}) \mathrm{\, mod \, } l 
    \end{align}
    It follows that 
    $
    -qh \equiv h(q^{2j+2}) \mathrm{\,mod\, } l$ and hence
    $
    h\equiv hq^{2(2j+1)} \mathrm{\,mod\,} l$.
    Since $l$ is odd prime and $0< h < l$, we have $q^{2(2j+1)}\equiv 1 \mathrm{\,mod\, } l.$   Hence,  $v_l | (2j+1)$ which implies that  $v_l$ is odd.
    It is not difficult to verify that $\textrm{ord}_l{(q)}\in \{v_l,2v_l\}$. Suppose that    $\textrm{ord}_l{(q)}=v_l$.  Since $v_l | (2j+1)$,  we have $h \equiv q^{2j+1}h \mathrm{\,mod\,}l$. 
    From \eqref{th1e1},  it can be concluded that
    $
    h(-q^{2j+1}) \equiv hq^{2k+1} \mathrm{\,mod\,}l$. It follows that 
    $   2h q^{2j+1}  \equiv 0 \mathrm{\, mod\,}l$. Since  $l\nmid 2 q^{2j+1}$, we have  $l|h$.  This is a 
    contradiction because  $0<h<l$.
    Therefore, $\textrm{ord}_l{(q)}=2v_l$ as desired.
    
    Conversely, assume that $\textrm{ord}_l{(q^2)}=v_l$  is odd and $\textrm{ord}_l{(q)}=2v_l$.  Let $0\leq h<l$. Then $
    q^{v_l}\equiv -1 \mathrm{\, mod  \,}l$ which implies that 
    $hq^{v_l}\equiv -h \mathrm{\, mod  \,}l$ and $
    hq^{v_l+1}\equiv -hq \mathrm{\, mod  \,}l$. Sine $v_l$ is odd,  $v_l+1$ is even and 
    $ hq^{2(\frac{v_l+1}{2})}\equiv -hq \mathrm{\, mod  \,}l$.
    Consequently,   $Cl_{q^2,l}(h) =Cl_{q^2,l}(-qh) $.
    By Lemma \ref{thm -qi},  all monic irreducible factors of $x^l-1$ over $\mathbb{F}_{q^2}$ are SCRIM as desired. 
\end{proof}

\begin{prop}  \label{cor2.9}
    Let $l$ be an odd prime  coprime  to $q$. Then either  all monic irreducible factors of $x^{l}-1$ over $\mathbb{F}_{q^2}$ are SCRIM or  $x-1$ is the only SCRIM factor of $x^{l}-1$ over $\mathbb{F}_{q^2}$.
\end{prop}

\begin{proof} Assume that   $x^l-1$  contains  more than one SCRIM factor in $\mathbb{F}_{q^2}[x]$. By Lemma \ref{thm -qi},   there exists an integer  $0< h < l$ such that  $Cl_{q^2,l}(h) =Cl_{q^2,l}(-qh)$.  Then there exists an integer  $j \geq 0$ such that $  h\equiv (-qh)(q^{2j}) \mathrm{\, mod \, } l$. Hence, 
    $
    h(1+q^{2j+1})\equiv 0 \mathrm{\, mod \, } l$. 
    Since $l$ is prime and $0< h < l$,  it follows that $(1+q^{2j+1})\equiv 0 \mathrm{\,mod\, } l$. Equivalently,  $-q^{2j+1}\equiv 1 \mathrm{\,mod\, } l$.
    Thus,  $
    i\equiv (-qi)q^{2j} \mathrm{\, mod\,}l$ for all $0\leq i <l$.
    Therefore,  $Cl_{q^2,l}(i)  =Cl_{q^2,l}(-qi)$ for all $0\leq i <l$.
    Hence,  all monic irreducible factors of $x^l-1$ over $\mathbb{F}_{q^2}$ are SCRIM.
\end{proof}

\begin{cor} \label{cor-1-1}
    Let $n$ be an odd positive integer coprime to $q$. Then  $x-1$ is the only SCRIM factor  of $x^{n}-1$ over $\mathbb{F}_{q^2}$  if and only if  for each prime divisor $l$ of $n$,  $\textrm{ord}_l{(q^2)}$ is even or     $\textrm{ord}_l{(q)}$ is odd.
\end{cor}
\begin{proof}   Assume  that there exists a prime divisor $l$ of $n$ such that $\textrm{ord}_l{(q^2)}$ is odd and   $\textrm{ord}_l{(q)}$ is even. Then by Theorem \ref{order}, all monic irreducible factors of $x^{l}-1$ over $\mathbb{F}_{q^2}$ are SCRIM. Since $(x^l-1)|(x^n-1)$,   $x^n-1$ contains more than one SCRIM factor over~$\mathbb{F}_{q^2}$.
    
    Conversely,  
    assume  that for each prime divisor $l$ of $n$, $\textrm{ord}_l(q^2)$ is even or $\textrm{ord}_l(q)$ is odd.  By Theorem  \ref{order} and Proposition \ref{cor2.9},  $x-1$ is the only   SCRIM factor of $x^l-1$ over $\mathbb{F}_{q^2}$. Suppose that  $x^{n}-1$ contains more than one SCRIM factor over $\mathbb{F}_{q^2}$.  By Lemma \ref{thm -qi}, there exists an integer $0<i<n$  such that $Cl_{q^2,n}(i)  =Cl_{q^2,n}(-qi)$.   
    Then  there exists a prime divisor $l$ of $n$ such that   $Cl_{q^2,l}(0 ) \ne Cl_{q^2,l}(i)  =Cl_{q^2,l}(-qi)$ which implies that  $x^{l}-1$ contains more than one SCRIM factor over $\mathbb{F}_{q^2}$, a contradiction.
    Therefore, $x-1$ is the only  SCRIM factor of $x^n-1$ over $\mathbb{F}_{q^2}$ as desired. 
\end{proof}

The following lemmas from \cite{M1997} are useful.

\begin{lemma}
    [{\cite[Proposition 4]{M1997}}]\label{ord} Let $q$ be a prime power and let $l$ be an odd prime such that $l\nmid q$.  Let  $r$ be a positive integer. Then ${\rm ord}_{l^r}(q)={\rm ord}_{l}( q)p^i$ for some $i\geq 0$. 
\end{lemma}

\begin{lemma}
    [{\cite[Lemma 1]{M1997}}]\label{congr} Let $a_1,a_2,\dots,a_t$ be positive integers. Then the system of congruences 
    \begin{align*}
    x\equiv a_1\,({\rm mod}\, 2a_1), \quad x\equiv a_2\,({\rm mod}\, 2a_2), \quad \dots, \quad x\equiv a_t\,({\rm mod}\, 2a_t) 
    \end{align*}
    has a solution $x$ if and only if there exists $s\geq 0$ such that $2^s|| a_i$ for all $i=1,2,\dots, t$. 
\end{lemma}

Some characterizations of SCRIM factors of  $x^{n}-1$ over $\mathbb{F}_{q^2}$ are given in the next theorem. 
\begin{thm}\label{char-n}
    Let $n$ be an odd   integer coprime to $q$. Then the following statements are equivalent.
    \begin{enumerate}[$1)$]
        \item  All monic irreducible factors of $x^{n}-1$ over $\mathbb{F}_{q^2}$ are SCRIM.
        \item  For each  prime divisor $l$  of $n$,  all monic irreducible factors of $x^{l}-1$ over $\mathbb{F}_{q^2}$ are SCRIM.
        \item  For each  prime divisor $l$  of $n$, $\textrm{ord}_l{(q^2)}$ is odd but    $\textrm{ord}_l{(q)}$ is even.
    \end{enumerate} 
\end{thm}
\begin{proof}
    To prove $1) \Rightarrow 2)$, assume that all monic irreducible factors of $x^{n}-1$ over $\mathbb{F}_{q^2}$ are SCRIM. Let $l$ be a prime divisor of $n$.   Since $(x^l-1) |(x^n-1)$,  all monic irreducible factors of $x^{l}-1$ over $\mathbb{F}_{q^2}$ are SCRIM as desired.
    
    From Theorem \ref{order}, we have   $2) \Rightarrow 3$.
    
    To prove $3)  \Rightarrow 1)$, assume that  $\textrm{ord}_l{(q^2)}$ is odd but   $\textrm{ord}_l{(q)}$ is even for all  prime divisors $l$  of $n$. Since $\textrm{ord}_l{(q)} \in \{ \textrm{ord}_l{(q^2)}, 2\textrm{ord}_l{(q^2)}\}$, we have   $2|| {\rm ord}_{l}(q)$. 
    Assume that $n=l_1^{r_1}l_2^{r_2}\dots l_t^{r_t}$, where $l_1$, $l_2$, $\dots$, $l_t$ are distinct primes and $r_i\geq 1$ for all $1\leq i\leq t$.  
    Then  $2|| {\rm ord}_{l_i}(q)$ for every prime $l_i$ dividing $n$. Then, by Lemma~\ref{ord}, we have $2|| {\rm ord}_{l_i^{r_i}}(q)$ for all $i=1,2,\dots,t$. By Lemma~\ref{congr}, there exists an integer $c$ such that
    $c\equiv \frac{{\rm ord}_{l_i^{r_i}}( q)}{2} \,({\rm mod}\, {\rm ord}_{l_i^{r_i}}( q)) $
    for all $i=1,2,\dots,t$. Since $\frac{{\rm ord}_{l_i^{r_i}}( q)}{2}$ is odd, $c$ is an odd integer. Thus $q^c\equiv -1 \,({\rm mod}\, l_i^{r_i})$ for all $i=1,2,\dots,t $ which implies that  $q^c\equiv -1 \,({\rm mod}\, n)$.   Then $c+1$ is odd and   $\left({q^2}\right)^{\frac{c+1}{2}}\equiv  q^{c+1} \equiv -q \,({\rm mod}\, n)$.  Consequently,  $Cl_{q^2,n}(h)=Cl_{q^2,n}(-qh) $ for all $0\leq h<n$.
    By Lemma \ref{thm -qi}, every  monic irreducible factor of $x^{n}-1$ over $\mathbb{F}_{q^2}$ is SCRIM.
\end{proof}

\subsection{Enumeration of SCRIM Factors of $x^{n}-1$ over $\mathbb{F}_{q^2}$}

First, note that the number of   SCRIM factors of $x^n-1$  over $\mathbb{F}_{q^2}$ can be  determined using \eqref{formulaO}. However, the said formula is not easily computed. The purpose of this subsection is to determined recursive formulas for $|\Omega_{q^2,n}|$ which may reduce the complexity in the computation. 

First,  we focus on the case where $n$ is even and determine the number  $|\Omega_{q^2,n}|$ of CRIM factors of $x^{n}-1$ over $\mathbb{F}_{q^2} $  in terms of $|\Omega_{q^2,n^\prime}|$ and  $m$, where $n^\prime $ is the odd positive integer such that $n=2^mn^\prime$.

\begin{thm} \label{evenInt}
    Let $n$ be a positive integer relatively prime to $q$ such that $n=2^mn'$ where $n'$ is odd and $m\geq 0$.  Let $r$ be the positive integer such that $2^r||(q+1)$. Then the number of SCRIM factors of $x^{n}-1$ over $\mathbb{F}_{q^2}$ is
    \begin{displaymath}
    |\Omega_{q^2,n}|=\left \{\begin{array}{ll}
    2^m|\Omega_{q^2,n'} | & \textrm{if\,\,} 0\leq m \leq r,\\
    2^r |\Omega_{q^2,n'}| & \textrm{if\,\,} m>r.
    \end{array}\right.
    \end{displaymath}
    In particular, if $n^\prime=1$, then  $|\Omega_{q^2,n'} | =1$ and     \begin{displaymath}
    |\Omega_{q^2,2^m}|=\left \{\begin{array}{ll}
    2^m & \textrm{if\,\,} 0\leq m \leq r,\\
    2^r& \textrm{if\,\,} m>r.
    \end{array}\right.
    \end{displaymath}
\end{thm}
\begin{proof} We consider the following two cases. 
    
    \noindent{\bf  Case 1:} $1\leq m \leq r$.
    Let $0\leq i<n'$ be  integer such that $Cl_{q^2,n^\prime}(i) =Cl_{q^2,n^\prime}(-qi)$.  Claim that   the  monic irreducible polynomials defined  corresponding to $Cl_{q^2,2^mn'}(i) , Cl_{q^2,2^mn'}(i+n') ,Cl_{q^2,2^mn'}(i+2n') ,\dots ,Cl_{q^2,2^mn'}(i+(2^m-1)n')$  are  the distinct  SCRIM factors  of $x^n-1$ over $\mathbb{F}_{q^2}$.  Consider the following steps.
    
    \noindent{\bf Step I}: Show that $Cl_{q^2,2^mn'}(i) =Cl_{q^2,2^mn'}(-qi) .$ 
    Since $Cl_{q^2,n'}(i) =Cl_{q^2,n'}(-qi) $, there exists an integer  $j\geq 0$ such that,
    \begin{align}\label{th2e1}
    i&\equiv (-qi)q^{2j}  \equiv  -i q^{2j+1}\textrm{\,mod \,} n'.
    \end{align}
    Since $2^m|2^r$ and $2^r||(q+1)$,  we have $2^m|(q+1)$. Precisely,  $q=2^mk-1$ for some positive integer $k$.  Then
    \begin{align*}
    q^{2j+1}+1&=(2^mk-1)^{2j+1}+1\\
    &=(2^mk-1)^{2j+1}+1^{2j+1}\\
    &=(2^mk-1 +1)\sum_{t=0}^ {2j}(-1)^t(2^mk-1)^{2j-t}\\
    &=2^mk\sum_{t=0}^ {2j}(-1)^t(2^mk-1)^{2j-t}
    \end{align*}
    It follows that     $ 1\equiv  -q^{2j+1} \, \textrm{mod}\,2^m$.
    Together with  \eqref{th2e1},  we have $i\equiv (-qi)q^{2j} \textrm{\,mod \,}2^mn'$. Consequently,   $Cl_{q^2,2^mn'}(i) =Cl_{q^2,2^mn'}(-qi) $.

    \noindent{\bf Step II}: Show that $Cl_{q^2,2^mn'}(i+kn') =Cl_{q^2,2^mn'}(-q(i+kn')) $ for all $1\leq k\leq 2^m-1.$ Let $1\leq k\leq 2^m-1$ be an integer. Since $ 1 +q^{2j+1} \equiv 0 \, \textrm{mod}\,2^m$ and $n'|kn'$,  we have $kn'(1+q^{2j+1}) \equiv 0\, \textrm{mod}\, 2^mn'$.  It follows that
    \begin{align*}
    (i+kn')+q^{2j+1}(i+kn')&\equiv  (1+q^{2j+1})(i+kn' )\textrm{\,mod \,} 2^mn'\\
    &\equiv i(1+q^{2j+1})+kn'(1+q^{2j+1}) \textrm{\,mod \,} 2^mn'\\
    &\equiv 0 \textrm{\,mod \,} 2^mn'
    \end{align*}
    and hence, 
    $
    (i+kn')\equiv -q(i+kn') q^{2j} \textrm{\,mod \,} 2^mn'$.  As desired,  we have 
    $Cl_{q^2,2^mn'}(i+kn' )=Cl_{q^2,2^mn'}(-q(i+kn')) 
    $ for all $1\leq k\leq 2^m-1$.

    \noindent{\bf Step III}: Show that $Cl_{q^2, 2^mn'}(i+kn') \neq Cl_{q^2,2^mn'}(i+hn') $ for all $ 0\leq k<h\leq 2^m-1$.
    Suppose  that there exist integers  $ 0\leq k<h\leq 2^m-1$ such that \[Cl_{q^2,2^mn'}(i+kn') =Cl_{q^2,2^mn'}(i+hn') .\]
    Then there exists an integer $j$ such that 
    $ i+kn'\equiv -q(i+hn')q^{2j} {\textrm{mod }2^mn'} $ which implies that   $ i+kn'\equiv -q(i+hn')q^{2j} {\textrm{mod }2^m} $.  It follows that 
    \begin{align} \label{eq-kh}
    i(1+q^{2j+1})+n'(k+hq^{2j+1})&\equiv 0{\textrm{ mod }2^m}.
    \end{align}
    Since $2^m|(1+q^{2j+1})$, we have  $2^m|i(1+q^{2j+1})$ and hence $2^m|n'(k+hq^{2j+1})$ by \eqref{eq-kh}.   Since $n'$ is odd,  it can be concluded that $2^m|(k+hq^{2j+1})$.   Since $k<h$, we have have $h=k+a$ for some positive integer $d\leq 2^m-1$. Hence, 
    
    \[k+hq^{2j+1}= k + (k+a) q^{2j+1}= k(q^{2j+1}+1) +aq^{2j+1}.\]
    Since $2^m|(k+hq^{2j+1}) $ and $2^m|(q^{2j+1}+1)$, it can be deduced that $2^m|aq^{2j+1}$. Since $q$ is odd, $2^m|a$, a contradiction. 
    Hence, the sets $Cl_{q^2,2^mn'}(i) , Cl_{q^2,2^mn'}(i+n') ,Cl_{q^2,2^mn'}(i+2n') ,\dots ,Cl_{q^2,2^mn'}(i+(2^m-1)n')$ are distinct.

    \noindent{\bf Step IV}: Show that for each $0\leq h <2^mn'$,  if $Cl_{q^2,2^mn'}(h) =Cl_{q^2,2^mn'}(-q h)$, then  $Cl_{q^2,2^mn'}(h) =Cl_{q^2,2^mn'}( i+kn' )$ for some integers $0\leq k\leq 2^m-1$ and  $0\leq i<n^\prime $ such that $Cl_{q^2,n'}(i) =Cl_{q^2,n'}(-q i)$.  Let  $0\leq h <2^mn'$ be such that  $Cl_{q^2,2^mn'}(h) =Cl_{q^2,2^mn'}(-q h)$. Then  $Cl_{q^2,2^mn'}(h) =  Cl_{q^2,2^mn'}(i+kn') =Cl_{q^2,2^mn'}(-q(i+kn'))  $ for some   for some integers $0\leq k\leq 2^m-1$ and   $ 0\leq i<n'$.  It follow that 
    $ i+kn'\equiv -q(i+kn')q^{2j}\textrm{mod  }2^mn'$ for some integer $j\geq 0$.   Hence, 
    \begin{align*}
    (i+iq^{2j+1})+kn'(1+q^{2j+1})&\equiv 0 \textrm{ mod  }2^mn'
    \end{align*}
    which implies that 
    \[(i+iq^{2j+1})+kn'(1+q^{2j+1})\equiv 0 \textrm{ mod  }n'.\]
    We    have    $i+iq^{2j+1} \equiv 0 \textrm{ mod  }n'$ and hence  $Cl_{q^2}(i)_{\textrm{mod }n'} =Cl_{q^2}(-qi)_{\textrm{mod }n'} $.

    Therefore, the  monic irreducible polynomials defined  corresponding to the cyclotomic cosets $Cl_{q^2,2^mn'}(i) , Cl_{q^2,2^mn'}(i+n') ,Cl_{q^2,2^mn'}(i+2n') ,\dots ,Cl_{q^2,2^mn'}(i+(2^m-1)n')$  are  the distinct  SCRIM factors  of $x^n-1$.  Hence,  $|\Omega_{q^2,n}| =
    2^m|\Omega_{q^2,n'} |$ as desired. 
    
    \noindent{\bf  Case 2:}  $m>r$. 
    Let  $k=m-r>0$. We prove that $|\Omega_{q^2,2^{m}n'}| =|\Omega_{q^2,2^{r+k}n'}| =
    2^r|\Omega_{q^2,n'} |$.   Since  $(x^{2^rn'}-1) | ({x^{2^{r+k}n'}-1})$,  it is suffices to show that every  SCRIM factor of ${x^{2^{r+k}n'}-1}$   is a divisor of $x^{2^rn'}-1$.  Suppose that there exists a SCRIM factor   $f(x)$   of ${x^{2^{r+k}n'}-1}$ such that  $f(x) \nmid { (x^{2^rn'}-1)}$. Then $f(x)\ne x-1$.   Let $\alpha$ be  a primitive  $2^{r+k}n'$th root of unity. Then $f(x)$  
    is the minimal polynomial of $\alpha^i$ for some $0< i<2^{r+k}n'$.       Suppose that  $i$ is even. Then $i=2i^\prime$ for some $i^\prime\in\mathbb{N}$ and
         $(\alpha^i)^{2^{r+k-1} n^\prime} =(\alpha^{2i^\prime})^{2^{r+k-1} n^\prime} =(\alpha^{i^\prime})^{2^{r+k} n^\prime}  =(\alpha^{2^{r+k} n^\prime} )^{i^\prime} =1$.  Applying the above expression recursively, we have 
           $(\alpha^i)^{2^r n^\prime} = (\alpha^i)^{2^{r+1}n^\prime} = \dots=(\alpha^i)^{2^{r+k} n^\prime}  =1$. 
         It follows that ${\rm ord}(\alpha^i)|2^rn^\prime$ which implies that  $f(x)|(x^{2^r n^\prime}-1)$, a contradiction. Hence, $i$ must be odd. 
    Since $f(x)$ is SCRIM,  $|Cl_{q^2,2^{r+k}n' }(i) |$ is  odd by \cite[Theorem 3.3]{BJU2015}.  It follows that $d:=\displaystyle \sum_{j\in Cl_{q^2,2^{r+1}n'}(i)} j $ is odd. 
    Since  \[  f(0)=\prod_{{j\in Cl_{q^2,2^{r+1}n'}(i)} }\alpha^j = \alpha^d,\] 
    it can be concluded that
    \[1=f(0)^{{\rm ord} _{q^2} (f(0))}=\left( \alpha^d\right)^{{\rm ord} _{q^2} (f(0))} .\]
    Since 
    $\alpha$ be  a primitive  $2^{r+k}n'$th root of unity,  we have $2^{r+k}n'| d {\rm ord} _{q^2} (f(0))$.  Since $d $ is odd, it is implies that 
    $\displaystyle 2^{r+k}| \textrm{ord}_{q^2}(f(0)) $.   
    
    Since $f(x)=f^{\dagger}(x)$, it follows that $ f(0) ^{-q} = f(0)  $. Equivalently,   $f(0)^{q+1} = 1$ which implies that $
    \textrm{ord}_{q^2}(f(0))|(q+1).$  Hence,  
    $2^{r+k}|(q+1)$, a contradiction. 
    Hence, there are no SCRIM polynomials $f(x)$ such that $f(x)\mid (x^{2^{r+k}n'}-1)$ but $f(x)\nmid( x^{2^rn'}-1)$. Therefore, $|\Omega_{2^mn'}|=|\Omega_{2^rn'}|.$

    In the case where  $n^\prime=1$, it is obvious that  $|\Omega_{q^2,n'} | =1$. Hence, the result follows.
\end{proof}

Next, we focus on the case where $n=n_1n_2$ is odd  and   $x-1$ is only the SCRIM factor  $x^{n_1}-1$ and $x^{n_2}-1$ over $\mathbb{F}_{q^2}$.

\begin{thm}\label{one-one-scrim-n}
    Let $n_1$ and $n_2$ be coprime odd integers relatively prime to $q$. If  $x-1$ is only the SCRIM factor  $x^{n_1}-1$ and $x^{n_2}-1$ over $\mathbb{F}_{q^2}$, then it is the only SCRIM factor of $x^{n_1n_2}-1$ over $\mathbb{F}_{q^2}$.  
\end{thm}
\begin{proof} 
    Let $0\leq h< n_1n_2$ such that $Cl_{q^2,n_1n_2}(h) =Cl_{q^2,n_1n_2}(-qh) $.  Then there exist $i \in \mathbb{N}$ such that $hq^{2i}\equiv -hq \mathrm{\,mod\,}  n_1n_2$. Since $\gcd(q,n_1)=1$ and $\gcd(q,n_2)=1$,  
    \begin{align*}
    hq^{2i}\equiv -hq \mathrm{\,mod\,} n_1 \textrm{~~~~and~~~~} hq^{2i}\equiv -hq \mathrm{\,mod\,} n_2.
    \end{align*}
    Since $x-1$ is only the SCRIM factor  $x^{n_1}-1$ and $x^{n_2}-1$ over $\mathbb{F}_{q^2}$,  we have $h\equiv 0 \mathrm{\,mod\,} n_1$ and $h\equiv 0 \mathrm{\,mod\,} n_2$. Since $n_1$ and $n_2$ are coprime,  we have $h=0$. Hence, $x^{ n_1n_2}-1$ contains only one  SCRIM over $\mathbb{F}_{q^2}$.
\end{proof}
The next corollary follows immediately. 
\begin{cor} 
    Let $n_1$ and $n_2$ be coprime odd integers relatively prime to $q$. If  $x-1$ is only the SCRIM factor  $x^{n_1}-1$ and $x^{n_2}-1$ over $\mathbb{F}_{q^2}$, then  the number of SCRIM factors of $x^{n_1n_2}-1$ over $\mathbb{F}_{q^2}$ is $|\Omega_{q^2,n_1n_2}|=|\Omega_{q^2,n_1}||\Omega_{q^2,n_2}|=1$.  
\end{cor}

In the following theorem, we focus on the case where $n=n_1n_2$ is odd  and  all irreducible factors of $x^{n_1}-1$ over $\mathbb{F}_{q^2} $ are SCRIM and  $x-1$ is the only SCRIM factor of  $x^{n_2}-1$ over $\mathbb{F}_{q^2}$.

\begin{thm}\label{all-one-scrim-n}
    Let $n_1$ and $n_2$ be coprime odd integers relatively prime to $q$. If all irreducible factors of $x^{n_1}-1$ over $\mathbb{F}_{q^2} $ are SCRIM and  $x-1$ is the only SCRIM factor of  $x^{n_2}-1$ over $\mathbb{F}_{q^2}$, then  the number of SCRIM factors   of $x^{n_1n_2}-1$ over $\mathbb{F}_{q^2}$  is $|\Omega_{q^2,n_1}|$. 
\end{thm}
\begin{proof}
    Assume that all irreducible factors of $x^{n_1}-1$ over $\mathbb{F}_{q^2}$ are SCRIM and $x-1$ is the only SCRIM factor of  $x^{n_2}-1$ over $\mathbb{F}_{q^2}$.  By Lemma \ref{thm -qi}, $Cl_{q^2,n_1}(h) =Cl_{q^2,n_1}(-qh) $ for all  $0\leq h<n_1$ and  $Cl_{q^2,n_2}(k) \ne Cl_{q^2,n_2}(-qk) $ for all  $0< k<n_2$.  It suffices to   prove the following steps.

    \noindent{\bf Step I}: Show that   $Cl_{q^2,n_1n_2}(h n_2) =Cl_{q^2,n_1n_2}(-qhn_2)$  for all  $0\leq h <n_1$.  Let  $0\leq h <n_1$.  Since $Cl_{q^2,n_1}(h) =Cl_{q^2,n_1}(-qh) $,  we have    
    $hq^{2i}\equiv -hq \mathrm{\,mod\,} n_1$ for some integer $i\geq 0$.
    Hence, $
    hn_2q^{2i}\equiv -hn_2q \mathrm{\,mod\,} n_1n_2$  which implies that  $Cl_{q^2,n_1n_2}(h n_2) =Cl_{q^2,n_1n_2}(-qhn_2)$  for all  $0\leq h <n_1$.

    \noindent{\bf Step II}: Show that  $Cl_{q^2,n_1n_2}(an_2)  \ne Cl_{q^2,n_1n_2}(bn_2)$  for all  $0\leq a <b<n_1$ such that  $Cl_{q^2,n_1}(a)  \ne Cl_{q^2,n_1}(b)$.  Suppose that  $Cl_{q^2,n_1n_2}(an_2)  =Cl_{q^2,n_1n_2}(bn_2)$  and  $Cl_{q^2,n_1}(a)  \ne Cl_{q^2,n_1}(b)$ for some  $0\leq a <b<n_1$.  Then $
    an_2q^{2i}\equiv bn_2 \mathrm{\,mod\,} n_1n_2$ for some integer $i\geq 0$.  It follows that  $
    aq^{2i}\equiv b \mathrm{\,mod\,} n_1$ which implies that  $Cl_{q^2,n_1}(a)  \ne Cl_{q^2,n_1}(b)$, a contradiction. Hence, the statement is proved.

    \noindent{\bf Step III}: Prove that   $Cl_{q^2,n_1n_2}(c) \ne Cl_{q^2,n_1n_2}(-qc)$ for all $c\notin \{an_2\mid 0\leq a <n_1\}$.  Let  $c\notin \{an_2\mid 0\leq a <n_1\}$. Then $c=an_2+b$ for some  $0\leq a<n_1$ and $0<b<n_2$. Suppose that $Cl_{q^2,n_1n_2}(c) = Cl_{q^2,n_1n_2}(-qc)$.  Then  $
    q^{2i}(an_2+b) \equiv -q(an_2+b) \mathrm{\,mod\,} n_1n_2$ for some integer $i\geq 0$.  It follows that     $
    q^{2i}(an_2+b) \equiv -q(an_2+b) \mathrm{\,mod\,} n_2$ and hence,     $
    q^{2i}b\equiv -qb\mathrm{\,mod\,} n_2$. It can be deduced that $Cl_{q^2,n_2}(b) = Cl_{q^2,n_2}(-qb)$, a contradiction. Hence, $Cl_{q^2,n_1n_2}(c) \ne Cl_{q^2,n_1n_2}(-qc)$ for all $c\notin \{an_2\mid 0\leq a <n_1\}$.

    Therefore, we have  $|\Omega_{q^2,n_1n_2}|=|\Omega_{q^2,n_1}|$  as desired. 
    %
    %
    %
    %
    %
    %
    %
\end{proof}

In the next theorem, we consider the case where $n=n_1n_2$ is odd  and  all  irreducible monic factors of $x^{n_1}-1$ and $x^{n_2}-1$ over $\mathbb{F}_{q^2}$ are SCRIM.

\begin{thm}\label{all-all-scrim-n}
    Let $n_1$ and $n_2$ be  odd integers coprime  to $q$. If all  irreducible monic factors of $x^{n_1}-1$ and $x^{n_2}-1$ over $\mathbb{F}_{q^2}$ are SCRIM, then all irreducible monic  factor of $x^{n_1n_2}-1$ over $\mathbb{F}_{q^2}$ are SCRIM. 
\end{thm}
\begin{proof} Assume that the  irreducible monic factors of $x^{n_1}-1$ and $x^{n_2}-1$ over $\mathbb{F}_{q^2}$ are SCRIM.  
    From Theorem \ref{char-n},  $\textrm{ord}_l{(q^2)}$ is odd and   $\textrm{ord}_l{(q)}$ is even for all  prime divisors $l$  of $n_1$ and $n_2$. Hence, $\textrm{ord}_l{(q^2)}$ is odd and   $\textrm{ord}_l{(q)}$ is even for all  prime divisors $l$  of $n_1n_2$.   By Theorem \ref{char-n}, the irreducible monic  factor of $x^{n_1n_2}-1$ over $\mathbb{F}_{q^2}$ are SCRIM. 
\end{proof}

Theorem \ref{all-all-scrim-n}  does not imply that  $|\Omega_{q^2,n_1n_2}|=|\Omega_{q^2,n_1}||\Omega_{q^2,n_2}|$.  However, it is true   in some cases.  It is not difficult to see that  $|\Omega_{5^2,7\cdot 23}|=9=3\cdot 3=|\Omega_{5^2,7}||\Omega_{5^2,23}|$  and  $|\Omega_{3^2,7\cdot 19}|=17\ne 3\cdot 3=|\Omega_{3^2,7}||\Omega_{3^2,19}|$. 
Necessary conditions for $n_1$ and $n_2$  to have $|\Omega_{q^2,n_1n_2}|=|\Omega_{q^2,n_1}||\Omega_{q^2,n_2}|$ are  given in the next corollary. 
\begin{cor} 
    Let $l_1$ and $l_2$ be  distinct odd primes relatively prime to $q$. If all irreducible factors of $x^{l_1}-1$ and $x^{l_2}-1$ over $\mathbb{F}_{q^2}$ are SCRIM   such  that $\gcd({\rm ord}_{l_1}(q^2),{\rm ord}_{l_2}(q^2))=1$, then the number of SCRIM factors of $x^{l_1l_2}-1$ over $\mathbb{F}_{q^2}$ is $|\Omega_{q^2,l_1l_2}|=|\Omega_{q^2,n_1}||\Omega_{q^2,n_2}|$. 
\end{cor}
\begin{proof} 
    From Theorem \ref{all-all-scrim-n}, all irreducible factors of $x^{l_1l_2}-1$ over $\mathbb{F}_{q^2}$ are SCRIM.  Since $\gcd({\rm ord}_{l_1}(q^2),{\rm ord}_{l_2}(q^2))=1$,  we have ${\rm ord}_{l_1l_2}(q^2)= {{\rm ord}_{l_1}(q^2){\rm ord}_{l_2}(q^2)}$.  Hence, 
    \begin{align*}
    |\Omega_{q^2,l_1l_2}|&=\displaystyle\sum_{d|l_1l_2}\frac{\phi(d)}{{\rm ord}_d(q^2)}\\
    &=1+\frac{\phi(l_1)}{{\rm ord}_{l_1}(q^2)}+\frac{\phi(l_2)}{{\rm ord}_{l_2}(q^2)}+\frac{\phi(l_1l_2)}{{\rm ord}_{l_1l_2}(q^2)}\\
    &=1+\frac{\phi(l_1)}{{\rm ord}_{l_1}(q^2)}+\frac{\phi(l_2)}{{\rm ord}_{l_2}(q^2)}+\frac{\phi(l_1)\phi(l_2)}{{\rm ord}_{l_1}(q^2){\rm ord}_{l_2}(q^2)}\\
    &=[1+\frac{\phi(l_1)}{{\rm ord}_{l_1}(q^2)}]+\frac{\phi(l_2)}{{\rm ord}_{l_2}(q^2)}[1+\frac{\phi(l_1)}{{\rm ord}_{l_1}(q^2)}]\\
    &=[1+\frac{\phi(l_1)}{{\rm ord}_{l_1}(q^2)}][1+\frac{\phi(l_2)}{{\rm ord}_{l_2}(q^2)}]\\
    &=|\Omega_{q^2,l_1}||\Omega_{q^2,l_2}|
    \end{align*} 
    as desired. 
\end{proof}

The results can be summarized in the following theorem. 

\begin{thm} \label{thmSUM}
    Let $q$ be a prime power and let $l_1,l_2, \dots ,l_t$ be distinct  odd primes relatively prime to $q$.   For each $i1\leq i\leq t$, let $r_i$ be a positive integer and let  $\textrm{ord}_{l_i}(q)=2^{a_i}d_i$, where $a_i\geq 0 $ is an integer and $d_i\geq 1$ is an odd   integer.  Then the following statements hold.
    \begin{enumerate}[$1)$]
        \item If  there exists $j\in \{1,2,\dots, t\}$ such that $a_j=0 $ or $a_j\geq  2$, then \[|\Omega_{q^2,\prod_{i=1}^{t}l_i^{r_i}}|= |\Omega_{q^2, \prod_{i=1}^{j-1}l_i^{r_i}\prod_{i=j+1}^{t}l_i^{r_i}}|.\] 
        \item If $a_1=a_2=\dots =a_t=1$, then \[ |\Omega_{q^2, \prod_{i=1}^{t}l_i^{r_i}}|=\displaystyle \sum_{d|\prod_{i=1}^{t}l_i^{r_i}}\dfrac{\phi(d)}{\textrm{ord}_d(q^2)}.\] 
        \item If  $a_i  \ne 1$ for all  $i\in \{1,2,\dots, t\}$, then \[|\Omega_{q^2, \prod_{i=1}^{t}l_i^{r_i}}|=1.\]
    \end{enumerate} 
    The empty product will be regarded as $1$.
\end{thm} 
\begin{proof} To prove $1)$,  assume that here exists $j\in \{1,2,\dots, t\}$ such that $a_j=0 $ or $a_j\geq  2$.  Then $\textrm{ord}_{l_i}(q)=d_i$  is odd or  $\textrm{ord}_{l_i}(q^2)$ is even.  Then the result can be obtained  via  Corollary  \ref{cor-1-1} and Theorem  \ref{all-one-scrim-n}.
    
    To prove $2)$, assume that $a_1=a_2=\dots =a_t=1$. Then $\textrm{ord}_{l_i}(q)=2d_i$ and  $\textrm{ord}_{l_i}(q^2)=d_i$ for all $1\leq i\leq t$.  The the result follows from   Theorem \ref{char-n} and \eqref{formulaO}. 
    
    To prove $3)$, assume that  $a_i  \ne 1$ for all  $i\in \{1,2,\dots, t\}$. Then for each $i\in\{1,2,\dots, t\} $ we have  $\textrm{ord}_{l_i}(q)=d_i$ or  $\textrm{ord}_{l_i}(q^2)$ is even.   Then the result follows from      Corollary  \ref{cor-1-1}  and Theorem \ref{one-one-scrim-n}.
\end{proof}

From Theorems \ref{evenInt} and  \ref{thmSUM},  the number  $|\Omega_{q^2,n}|$ can be determined for all prime power $q$ and  all positive integers $n$ with $\gcd(n,q)=1$ using the following steps. 
\begin{enumerate}
    \item  Write $n=2^mn^\prime$, where $n^\prime $ is an odd positive integer and $m\geq 0$ is an integer. 
    
    \item  Consider the following $2$ cases. 
    \begin{enumerate}
        \item $n^\prime =1$. Then $|\Omega_{q^2,n}|=|\Omega_{q^2,2^m}|$ can be determined in  Theorem \ref{evenInt}.
        \item $n^\prime \geq 3$. Compute a prime factorization of    $n^\prime$. Then $n^\prime=\prod_{i=1}^{t}l_i^{r_i}$, where    $t$, $r_1,r_2,\dots, r_t$  are   positive integers and  $l_1,l_2,\dots, l_t$ are distinct primes. 
        \begin{enumerate}
            \item Compute $\textrm{ord}_{l_i}(q)$. 
            
            \item  Set $n^{\prime\prime }= \prod\limits_{i=1, 2||\textrm{ord}_{l_i}(q) }^{t}l_i^{r_i}$. 
            
            \item  Compute $|\Omega_{q^2,n^\prime}|=|\Omega_{q^2,n^{\prime\prime}}|$ using  Theorem   \ref{thmSUM}.
            
            \item  Compute $|\Omega_{q^2,n}|= |\Omega_{q^2,2^m}||\Omega_{q^2,n^\prime}| = |\Omega_{q^2,2^m}||\Omega_{q^2,n^{\prime\prime}}| $  using   Theorems \ref{evenInt} and   \ref{thmSUM}.
        \end{enumerate}
    \end{enumerate}

\end{enumerate}

\section{Applications}

Due to their rich algebraic structures and wide applications, a family of cyclic codes has become of interest and been continuously studied. 
In this section, discussion on applications of SCRIM factors of $x^n-1$ over $\mathbb{F}_{q^2}$ in  the study of some   cyclic codes  will be  given.    In Subsection 3.1, the characterization and  enumeration of Hermitian complementary dual codes of arbitrary lengths  are given in terms of SCRIM factors of $x^n-1$ over $\mathbb{F}_{q^2}$. Note that some results on  Hermitian self-dual simple root  cyclic codes over finite fields have been given in \cite{JLS2014}.   The characterization and enumeration  of Hermitian self-dual simple root cyclic codes  over finite chain rings of prime characteristic  are given based on the SCRIM factors of $x^n-1$ over $\mathbb{F}_{q^2}$ in Subsection 3.2.

First, some basic concepts of cyclic codes over finite fields and over finite rings are recalled.  In this paper, we focus on codes over a  finite field $\mathbb{F}_{q^2}$ and  a  finite chain ring  $\displaystyle R:=\mathbb{F}_{q^2}+u\mathbb{F}_{q^2}+\dots+ u^{t-1}\mathbb{F}_{q^2},$ where  $t\geq 1$ and $u^t=0$.  Note that $R\cong \mathbb{F}_{q^2}$if $t=1$. The ring $R$  can be viewed as a quotient ring  $\displaystyle R\cong \mathbb{F}_{q^2}[z]/\langle z^t \rangle $. Note that    $R$ is local with maximal ideal $\langle u \rangle$, nilpotency index $t$,  and order $q^{2t}$. The residue  field of $R$ is   $\mathbb{F}_{q^2}$   and the characteristic of $R$ equals the characteristic of $\mathbb{F}_{q^2}$ which is prime.  Let $\varphi: R\to  \mathbb{F}_{q^2} \cong R/\langle u\rangle $  be defined by $\varphi(a)=a+\langle u\rangle $ for all $a\in R$. Extend the map to be $\varphi: R[x]\to \mathbb{F}_{q^2}[x]$ be defined by \[\sum_{i=0}a_ix^i\mapsto   \sum_{i=0}\varphi(a_i)x^i.\]

A \textit{linear code} of length $n$ over~$R$  is an $R$-submodule of the $R$-module~$R^n$.  In particular, if $t=1$,  $R^n=\mathbb{F}_{q^2}^n$ is an $\mathbb{F}_{q^2}$-vector space and a linear code over $R$ becomes a $\mathbb{F}_{q^2}$-subspace of  $\mathbb{F}_{q^2}^n$. A  linear code $C$ of length $n$ over~$R$ is called a  \textit{cyclic code} if it is  invariant under the cyclic shift. Precisely, $(c_{n-1},c_0,\dots,c_{n-2})\in C$ for all $(c_{0},c_1,\dots,c_{n-1})\in C$.

For
$\boldsymbol{u}=(u_0,u_1,\ldots ,u_{n-1}) \text{~and~}\boldsymbol{v}=(v_0,v_1,\ldots ,v_{n-1})$
in $R^n$, the \textit{Euclidean inner product}  of  $\boldsymbol{u}$ and $\boldsymbol{v}$ is defined to be 
\begin{align*}
\langle\boldsymbol{u},\boldsymbol{v}\rangle _{\rm E}:=\displaystyle\sum_{i=0}^{n-1}u_iv_i
\end{align*}
and the \textit{Hermitian inner product}  of~$\boldsymbol{u}$ and $\boldsymbol{v}$ is defined to be 
\begin{align*}
\langle\boldsymbol{u},\boldsymbol{v}\rangle _{\rm H}:=\displaystyle\sum_{i=0}^{n-1}u_i\bar{v_i}
\end{align*}
where $\bar{~}$ is the conjugation 
\begin{align} \label{eqCong} \displaystyle a=\sum_{i=0}^{n-1}u^ia_i\mapsto \sum_{i=0}^{n-1}u^ia_i^{q}
\end{align} for all $a\in R$. 

The \textit{Euclidean dual} of a linear code $C$ in $R^{n}$ is defined to be the set
\[C^{\perp _{\rm E}}:= \{\boldsymbol{u}\in R^{n} \mid  \langle\boldsymbol{u},\boldsymbol{v}\rangle _{\rm E}=0 \text{ for all } \boldsymbol{v}\in C\}\]
and the \textit{Hermitian dual} of a linear code $C$ in $R^{n}$ is defined to be the set
\[C^{\perp _{\rm H}}:= \{\boldsymbol{u}\in R^{n} \mid  \langle\boldsymbol{u},\boldsymbol{v}\rangle _{\rm H}=0 \text{ for all } \boldsymbol{v}\in C\}.\]
A linear code  $C$ over $R$  is said to be \textit{Euclidean self-dual} (resp. \textit{Hermitian self-dual}) if $C=C^{\bot_{\rm E}}$ (resp. $C=C^{\bot_{\rm H}}$). A linear code  $C$ over $R$  is said to be \textit{Euclidean complementary dual} (resp. \textit{Hermitian  complementary dual}) if $C\cap C^{\bot_{\rm E}}=\{0\}$ (resp. $C\cap C^{\bot_{\rm H}}=\{0\}$).

It is well known  that  the Euclidean and Hermitian dual of a cyclic  code $C$ is again a cyclic code over~$R$ and a  cyclic code   of length~$n$ over  $R$  can be viewed as an isomorphic   ideal of the quotient ring~$R[x]/\langle x^{n}-1\rangle $. In order to characterize the ideals in $R[x]/\langle x^{n}-1\rangle $, the following basic concepts   can be generalized from  that of polynomials over finite fields in Section 1.  For a  polynomial $f(x)=\sum_{i=0}^k f_ix^i$  of degree $k$ in $R[x]$ whose constant term  $f_0$ is a unit in $R$,  let $f^*(x)$ denote the {\em reciprocal polynomial} of $f(x)$, i.e.,  $f^*(x)= x^kf_0^{-1}f\left(\frac{1}{x}\right)$.  A polynomial $f(x)$ is said to be {\em self-reciprocal} if $f(x)=f^*(x)$. The {\em   conjugate polynomial} of $f(x)$ in $R[x]$ is defined to be $\overline{f(x)}=   \overline{f_0}+\overline{f_1}x+\dots+\overline{f_n}x^n$, where $\bar{~}$ is the conjugation in $R$ defined in \eqref{eqCong}. The {\em conjugate-reciprocal polynomial} of $f(x)$  is defined to be   $f^\dagger(x)= \overline{f^*(x)}$. A polynomial $f(x)$ is said to be {\em self-conjugate-reciprocal} if $f(x)=f^\dagger(x)$.

SCRIM factors of $x^n-1$ over $\mathbb{F}_{q^2}$  can be applied in the study of Hermitian complementary dual cyclic codes over  $\mathbb{F}_{q^2}$ and  of Hermitian self-dual simple root  cyclic codes over  $R$ in the following subsections. 

\subsection{Hermitian Complementary Dual Cyclic Codes over Finite Fields}

First, we note that the characterization and enumeration of       Hermitian complementary dual cyclic codes over finite fields  can be viewed as the special case of Hermitian complementary dual abelian codes in group algebras \cite{BJU2018}, where the underlying group is cyclic. Here, simpler and direct study of such complementary dual codes of arbitrary lengths  using SCRIM factors of $x^n-1$ in $\mathbb{F}_{q^2}[x]$. Some results on  Hermitian self-dual simple root  cyclic codes over finite fields have been discussed in \cite{JLS2014}.

Recall that a cyclic code $C$ of length $n$ over $\mathbb{F}_{q^2}$ can be viewed as an ideal in the principal ideal ring  $\mathbb{F}_{q^2}[x]/\langle x^n-1\rangle$ generated by a unique monic divisor of $x^n-1$ (see \cite{LC2004}). The such polynomial is called the {\em generator polynomial} for $C$. The generator polynomial of the Hermitian dual $C^{\perp_{\rm H}}$ can be determined in the next proposition.

\begin{prop}[{\cite{SJLU2015}}] Let $C$ be a  cyclic code of length $n$ over $\mathbb{F}_{q^2}$ with generator polynomial $g(x)$. Then $C^{\perp_{\rm H}}$ is a cyclic code of length $n$ over $\mathbb{F}_{q^2}$ with generator polynomial $h^\dagger(x)$, where 
    $h(x)=\frac{x^n-1}{g(x)}$.
\end{prop}

Necessary and sufficient conditions for cyclic codes of length $n$  over $\mathbb{F}_{q^2}$ to be Hermitian complementary dual are given in the next theorem.

\begin{thm}\label{charHLCD} Let $\mathbb{F}_{q^2}$ be a finite field   and let $n $ be a positive integer.   Let $C$ be a  cyclic code of length $n$ over $\mathbb{F}_{q^2}$ with generator polynomial $g(x)$. Then the following statements are equivalent.
    \begin{enumerate}[$1)$]
        \item  $C$ is Hermitian complementary dual.
        \item  $\gcd(g(x),h^\dagger(x))=1$.
        \item $g(x)$ is self-conjugate-reciprocal and for each irreducible factor of $g(x)$ has the same multiplicity  in $g(x)$ and in $x^n-1$. 
    \end{enumerate}
    
\end{thm}
\begin{proof}   Note that  $\langle \gcd(g(x),h^\dagger(x))\rangle =C+C^{\perp_{\mathrm{H}}} $.   Then $C$ is Hermitian complementary dual if and only if $C+C^{\perp_{\mathrm{H}}} = \mathbb{F}_{q^2}^n= \langle 1\rangle$.  Equivalently,   $\gcd(g(x),h^\dagger(x))=1$. Therefore, $1)$ and $2)$ are equivalent.

    To prove $2) \Rightarrow 3)$, assume that  $\gcd(g(x),h^\dagger(x))=1$.  Since $h^\dagger(x)|(x^n-1)$, we have  $g(x)h^\dagger(x)=x^n-1= (x^n-1)^\dagger=  g^\dagger (x)h^\dagger(x)$.  Hence, $g(x)$ is conjugate-self-reciprocal. Suppose that there exists a monic irreducible factor $f(x)$ of $g(x)$ such that the multiplicity in $g(x)$ is less than  the multiplicity in $x^n-1$.  Since $g(x)=g^\dagger(x) $,   $f(x)|g^\dagger(x)$ and $f(x)|h^\dagger(x)$. It follows  that $f(x)|\gcd(g(x),h^\dagger(x))$, a contradiction. Hence, every  irreducible factor of $g(x)$ has the same multiplicity  in $g(x)$ and in $x^n-1$.

    To prove $3) \Rightarrow 2)$,  assume that  $\gcd(g(x),h^\dagger(x))\ne 1$. Let $f(x)$ be a monic irreducible factor of  $\gcd(g(x),h^\dagger(x))$. Let $t$ be the multiplicity of $f(x)$ in $g(x)$.  Then $f(x)^t|g(x)$ and $f(x)|h^\dagger(x)$.
    Assume further that $g(x)=g^\dagger(x)$. Then $h^\dagger(x)=h(x)$ and $f(x)|h(x)$.  Hence, $f(x)^{t+1}$ is a divisor of $g(x)h(x)=x^n-1$. Precisely,  the multiplicity of $f(x)$ in   $g(x)$  is less than  the  multiplicity  of $f(x)$ in $x^n-1$. 
\end{proof}

The next corollary follows immediately from the previous theorem.
\begin{cor}[{\cite[Lemma 2]{L2017}}]
    Let $\mathbb{F}_{q^2}$ be a finite field   and let $n $ be a positive integer such that $\gcd(n,q)=1$.  Let $C$ be a  cyclic code of length $n$ over $\mathbb{F}_{q^2}$ with generator polynomial $g(x)$.  Then $C$ is Hermitian complementary dual if and only if  $g(x)=g^\dagger(x)$.   
\end{cor}

\begin{cor} Let $\mathbb{F}_{q^2}$ be  a finite field of characteristic $p$ and let $n=p^\nu n^\prime$, where $\nu\geq 0$ and $p\nmid n^\prime$. Then the number of Hermitian complementary dual cyclic codes of length $n$ over $\mathbb{F}_{q^2}$ is independent of $\nu$ and it is 
    \[2^{|\Omega_{q^2,n^\prime}|+|\Lambda_{q^2,n^\prime} |},\]
    where $|\Omega_{q^2,n^\prime}|$ and $|\Lambda_{q^2,n^\prime} |$ are determined in previous section.
\end{cor}
\begin{proof}
    From  Theorem \ref{charHLCD}, a cyclic code of length $n$ over $\mathbb{F}_{q^2}$ is Hermitian complementary dual if and only if  $C$ has generator polynomial of the following form 
    
    \begin{align*} 
    g(x)=\prod_{i=1}^{|\Omega_{q^2,n'}|}  f_i(x)^s\prod_{j=1}^{|\Lambda_{q^2,n'}|}\left(g_i(x)g_j^{\dagger} (x)\right)^t,
    \end{align*}
    where $s,t\in \{0,p^\nu\}$ and  \begin{align*}
    x^{n'}-1=\prod_{i=1}^{|\Omega_{q^2,n'}|}  f_i(x)\prod_{j=1}^{|\Lambda_{q^2,n'}|}\left(g_i(x)g_j^{\dagger} (x)\right)
    \end{align*}
    is the factorization in the form of  \eqref{eq-fistFactor}. Hence, the result follows.
\end{proof}

%

\subsection{Hermitian Self-Dual Cyclic Codes over Finite Chain Rings}

In this subsection, we focus on  Hermitian self-dual cyclic codes over a  finite chain ring $R=\mathbb{F}_{q^2}+u\mathbb{F}_{q^2}+\dots+u^{t-1}\mathbb{F}_{q^2}$.  For  $t=1$, we have $R=\mathbb{F}_{q^2}$ and  the study of such self-dual  codes has been given in \cite[Subsection III.B]{JLS2014}.  In this case,  there exists a Hermitian self-dual cyclic code of length $n$ over $R$   if and only if  $q$ is a $2$ power and $n$ is even.    The goals of this subsection are to characterize and enumerate  simple root  Hermitian self-dual cyclic codes of length $n$  over $R$, with $t\geq 2$.

Throughout,  we assume that $\gcd(n,q)=1$ and $t\geq 2$.   The characterization of cyclic codes over $R$ determined in  \cite{BSG2016} and \cite{HQDSRLP2004}.  Here, necessary results from \cite{BSG2016} are recalled.  Let $r_0=1+ua$, where $a$ is a unit in $R$. 
Based on the Hensel's Lemma,  there exists an index set $I\subseteq \mathbb{N}$ such that $\displaystyle x^n-1=\prod_{i \in I}h_i(x)$ is a  factorization of $x^n-1$ in to a product of distinct monic  irreducible polynomials in $\mathbb{F}_{q^2}[x]$ and   $\displaystyle x^n-r_0=\prod_{i \in I}f_i(x)$ is  a factorization of $x^n-r_0$ in to a product of  pairwise coprime monic basic irreducible polynomial  in $R[x]$ with $\varphi(f_i(x))=h_i(x)$.  From \cite[Notation 2.4 and Theorem 2.7 ]{BSG2016},  every cyclic code  $C$ of length $n$ over $R$ can be viewed as an ideal in $R[x]/\langle x^n-1\rangle $ generated by \begin{align} \label{gencyclic} \prod_{i\in I}f^{k_i}_i(x),\end{align} where $0\leq k_i\leq t$. In this case,  the sequence $(k_i) _{i\in I}$ is unique and $|C| =q^{ \sum_{i\in I}(t-k_i\mathrm{deg}{f_i(x)})}$.

Note that for each $i\in I$, $f_i^*(x)$ and $f_i^\dagger (x)$ are   monic basic irreducible factors of $x^n-r_0$ in $R[x]$. Then there exist permutations $\hat{ }$ and $'$ on $I$ such that   $f_{\hat{i}}(x)=f^{\dagger}_i(x)$  and  $f_{i^{'}}(x)=f_i^*(x)$. Consequently,  $f_{\hat{i}}(x)=f_i^{\dagger}(x)=\overline{(f^*_{i}(x))}=\overline{f_{i'}(x)}$. Moreover,    $h_{\hat{i}}(x)=h^{\dagger}_i(x)$ and   $h_{i^{'}}(x)=h_i^*(x)$. In general, we have the following fact. 

\begin{lemma}\label{enume}
    Let $f(x)\in R[x]$ be a monic irreducible factor  of $x^n-r_0$ such that $\varphi({f(x)})=h(x)\in \mathbb{F}_q[x]$. Then $f(x)=f^{\dagger}(x)$ if and only if $h(x)=h^{\dagger}(x).$
\end{lemma}
\begin{proof} It is not difficult to see that   $h(x)=h^{\dagger}(x)$ if  $f(x)=f^{\dagger}(x)$. 
    
    Conversely, assume that $h(x)=h^{\dagger}(x).$ Suppose $f(x)\neq f^{\dagger}(x)\in R[x].$  Then  they  are coprime in $R[x]$. Let   $\rho$ be  the   ring homomorphism defined 
    \begin{align*}
    \rho: R[x]/\langle f(x)\rangle \rightarrow \mathbb{F}_q[x]/\langle h(x)\rangle,\\
    \sum_{j=0}^{n}a_jx^j+\langle f(x)\rangle \rightarrow \sum_{j=0}^{n}\varphi({a_j})x^j+\langle h(x)\rangle.
    \end{align*}
    Since $f(x)$ and $f^{\dagger}(x)$ are coprime in $R[x]$, this implies that $f^{\dagger}(x)+\langle f(x)\rangle$ is a unit in $R[x]/\langle f(x)\rangle.$ Then $\rho (f^{\dagger}(x)+\langle f(x)\rangle)=h^{\dagger}(x)+\langle h(x)\rangle=\langle h(x)\rangle$, a contradiction. Hence, $f(x)\neq f^{\dagger}(x)\in R[x].$
\end{proof}

A relation between a cyclic code $C$ and its Hermitian dual  is  given based on {\cite[Lemma $3.2$]{BSG2016}} in the next lemma.

\begin{lemma} \label{lem3.7}
    Let $q$ be  a prime power and let $n$ be  a positive integer such that $\gcd(n,q)=1$.  Let $\displaystyle C=\langle \prod_{i\in I}f^{k_i}_i(x)\rangle$ be a cyclic code of length $n$ over $R$ with $0\leq k_i \leq t$ (see \eqref{gencyclic}). Then $\displaystyle C^{\perp_{\mathrm{H}}}=\langle \prod_{i \in I}f^{t-{k_i}}_{\hat{i}}(x)\rangle$ and $|C^{\perp_{\mathrm{H}}}|=q^{\displaystyle \sum_{i\in I}k_i\mathrm{deg}{f_i(x)}}$.
\end{lemma}

\begin{proof} By {\cite[Lemma $3.2$]{BSG2016}}, we have 
    \begin{align*}
    {C^{\perp_{\textrm{E}}}}= {\langle \prod_{i \in I}f^{t-{k_i}}_{{i'}}(x)\rangle},
    \end{align*}
    and hence, 
    \begin{align*}
    C^{\perp_{\textrm{H}}}=\overline{C^{\perp_{\textrm{E}}}}=\overline{\langle \prod_{i \in I}f^{t-{k_i}}_{{i'}}(x)\rangle}=\langle \prod_{i \in I}f^{t-{k_i}}_{\hat{i}}(x)\rangle.
    \end{align*}
    From  {\cite[Lemma $3.2$]{BSG2016}},  it follows that $ |C^{\perp_{\textrm{H}}}|=|C^{\perp_{\textrm{E}}}|=q^{\displaystyle \sum_{i\in I}k_i\textrm{deg}f_i(x)}$. 
\end{proof}

Necessary and sufficient conditions for a cyclic code $C$ over $R$ to be Hermitian self-dual are given in the next theorem. 
\begin{thm}\label{ki+ki'}
    Let $q$ be  a prime power and let $n$ be  a positive integer such that $\gcd(n,q)=1$. 
    Let $C=\langle \prod_{i\in I}f^{k_i}_i(x)\rangle$ be a cyclic code of length $n$ over $R$ (see \eqref{gencyclic}). Then $C$ is Hermitian self-dual if and only if $k_i+k_{\hat{i}}=t$ for all $i \in I$.
\end{thm}
\begin{proof} Recalled that $\hat{ }$  is a permutation  on $I$ such that   $ \widehat{\hat{i} }=i$.
    From Lemma \ref{lem3.7},  $C$ is Hermitian self-dual if and only if 
    \[
    \langle \prod_{i \in I}f^{{k_{{i}}}}_{i}(x)\rangle
    = C=C^{\perp_{\textrm{H}}}
    =\langle \prod_{i \in I}f^{t-{k_i}}_{\hat{i}}(x)\rangle
    =\langle \prod_{i  \in I}f^{t-{k_{\hat{i}}}}_{\widehat{\hat{i}}}(x)\rangle
   = \langle \prod_{i \in I}f^{t-{k_{\hat{i}}}}_{i}(x)\rangle .
   \]
    Equivalently,  $t-k_{\hat{i}}=k_i$   which implies that  $k_{\hat{i}}+k_i=t$ for all $i\in I$. 
\end{proof}

The characterization for the existence of a Hermitian self-dual cyclic codes of length $n$ over $R$ is given as follows. 
\begin{thm} \label{theorem3.10} 
    Let $q$ be  a prime power and let $n$ be  a positive integer such that $\gcd(n,q)=1$. There exists a Hermitian self-dual cyclic code of length $n$ over $R$ if and only if  the nilpotency index $t$ of $R$ is even. 
\end{thm}
\begin{proof}Assume that $\displaystyle C=\langle \prod_{i\in I}f^{k_i}_i(x)\rangle$ is  a Hermitian self-dual cyclic code of   length $n$ over $R$. Then by Theorem \ref{ki+ki'}, we have  $k_i+k_{\hat{i}}=t$ for all $i \in I$.  Without loss of generality,  assume that $f_1(x)$ is a polynomial suxh that $\varphi(f_1(x))=x-1$.   Since $ x-1$ is SCRIM, 	$f_1(x)=f_{\hat{1}} (x)$ by Lemma \ref{enume}. This  implies that  $2k_1=t$. Hence, $t$ is even.

    Conversely, assume that  $t$ is even. Then   $\displaystyle C=\langle \prod_{i\in I}f^{\frac{t}{2}}_i(x)\rangle$ is a Hermitian self-dual cyclic code of length $n$ over $R$.
\end{proof}


 Let   $I_1=\{i\in I|i=\hat{i}\}$ and $I_2=\{i\in I|i<\hat{i}\}$.  Since  $\hat{ }$  is a permutation  on $I$ such that   $ \widehat{\hat{i} }=i$,  the monic basic irreducible factors of  $x^n-r_0$ can be rearrange of the form $\displaystyle x^n-r_0=\prod_{i\in I_1}f_i(x)  \prod_{i\in I_2}f_{i}(x)f_{\hat{i}}(x)$. Consequently, $f_i(x)=f_i^\dagger(x)$ for all $i\in I_1$ and  $f_i(x)\ne f_i^\dagger(x)$ for all $i\in I_2$.  From  Lemma \ref{enume}, it follows that the number  $|I_1|$ of  self-conjugate-reciprocal monic basic irreducible factors of $x^n-r_0$  is equal to $|\Omega_{q^2,n^\prime}|$, the number of SCRIM factors of $x^n-1$ in $\mathbb{F}_q[x]$. Subsequently,     $|I_2|=|\Lambda_{q^2,n^\prime}|$.

 From Theorem  \ref{ki+ki'} and Theorem \ref{theorem3.10},  it can be deduced that a cyclic code    $C=\langle \prod_{i\in I_1}f_i^{k_i}(x)  \prod_{i\in I_2}f_{i}^{k_i}(x) f_{\hat{i}}^{k_{\hat{i}}} (x) \rangle $ of length $n$ over $R$ is  Hermitian self-dual if and only if $k_i=\frac{t}{2}$ for all $i\in I_1$ and $k_i+k_{\hat{i}}=t$ for all $i\in I_2$

The next corollary follows mediately for the above discussions.  
    
    \begin{cor}\label{numt}
        Let $t$ be an even positive integer. Then $C$ is a Hermitian self-dual cyclic code of length $n$ over $R$ if and only if $C$ can be written as a form $\langle \prod_{i\in I_1}f_i^{\frac{t}{2}}(x)  \prod_{i\in I_2}f_{i}^{k_i}(x) f_{\hat{i}}^{t-k_{{i}}} (x) \rangle ,$ where $0 \leq k_j \leq t.$ Then the number of Hermitian self-dual cyclic code of length $n$ over $R$ is $(t+1)^{|I_2|} =(t+1)^{|\Lambda_{q^2,n^\prime}|}$. 
    \end{cor}

\section{Conclusion and Remarks}
SCRIM  factors of $x^n-1$ over finite fields of square order have been  studied. The characterization of such factors is given together the enumeration formula in recursive forms.     Applications of obtained results in  coding theory have been discussed.   The characterization and enumeration of Hermitian complementary dual cyclic codes  of arbitrary lengths have been established. The characterization and enumeration of Hermitian  of  Hermitian self-dual simple root cyclic codes over finite chain rings of prime characteristic    have been  given  in terms of   SCRIM factors of $x^n-1$ over  $\mathbb{F}_{q^2}$.

It would be interesting to study SCRIM and SRIM factors of $x^n-\lambda$ over a finite field  and there applications in the study of some families of $\lambda$-constacyclic codes, where $\lambda$ is a unit in a finite field.

%


\end{document}